\documentclass{article}
\usepackage{sp_conf}

\usepackage[english]{babel}
\usepackage[utf8]{inputenc}
\usepackage[T1]{fontenc}
\usepackage{lmodern}
\usepackage{braket}
\usepackage{wrapfig}
\usepackage{graphicx}
\usepackage{amsmath,amssymb,amsfonts,verbatim}
\usepackage{hyperref}
\usepackage{mathtools}
\usepackage{enumitem}

\usepackage{xr}

\usepackage{amsmath,amssymb,amsfonts,verbatim}
\usepackage{algorithmic}
\usepackage{graphicx}
\usepackage[utf8]{inputenc}
\usepackage{textcomp}
\usepackage{xcolor}
\usepackage{braket}
\usepackage{soul}
\usepackage[colorinlistoftodos]{todonotes}
\usepackage{authblk}

\usepackage{caption}
\usepackage{subcaption}

\usepackage{enumitem}

\usepackage{cite}
\usepackage{amsmath,amssymb,amsfonts,verbatim}
\usepackage{amsthm}
\usepackage{algorithmic}
\usepackage{graphicx}
\usepackage[utf8]{inputenc}
\usepackage{textcomp}
\usepackage{xcolor}
\usepackage{braket}
\usepackage{soul}
\usepackage{enumitem}
\usepackage[colorinlistoftodos]{todonotes}

\usepackage{wrapfig}

\usepackage{caption}

\newtheorem{theorem}{Theorem}
\newtheorem{lemma}{Lemma}

\newtheorem{definition}{Definition}

\usepackage{tikz}
\usetikzlibrary{positioning}
 
\tikzstyle{pinstyle} = [pin edge={to-,thin,black}]
 
\allowdisplaybreaks
\newcommand{\nat}{\mathbb{N}}
\newcommand{\reals}{\mathbb{R}}
\newcommand{\posreals}{\mathbb{R}_+}
\newcommand{\consi}{\beta_t^i}
\newcommand{\feasconsi}{q_t^i}
\newcommand{\cons}{\beta_t}
\newcommand{\utili}{U^i}
\newcommand{\pricet}{\alpha_t}
\newcommand{\asconsi}{\hat{\beta}_t^i}
\newcommand{\dataset}{\mathcal{D}}
\newcommand{\parweighti}{\mu_t^i}

\newcommand{\gaussN}{\mathcal{N}}

\title{Identifying Coordination in a Cognitive Radar Network - A Multi-Objective Inverse Reinforcement Learning Approach
}

\name{Luke~Snow$^{\ \dagger}$,
        Vikram~Krishnamurthy$^{\ \dagger}$,
        and~Brian~M.~Sadler 
\sthanks{This research was funded by National Science Foundation grant CCF-2112457,  Army Research office grant W911NF-21-1-0093 , and Air Force Office of Scientific Research grant FA9550-22-1-0016.}}
\address{$^{\dagger}$ Electrical and Computer Engineering, Cornell University, Ithaca, NY\\
$^*$ DEVCOM Army Research Laboratory, Adelphi, MD}

\begin{document}

\maketitle

\begin{abstract}

Consider a target being tracked by a cognitive radar network. If the target can intercept some radar network emissions, how can it detect coordination among the radars? By 'coordination' we mean that the radar emissions satisfy Pareto optimality with respect to multi-objective optimization over each radar's utility. This paper provides a novel multi-objective inverse reinforcement learning approach which allows for both detection of such Pareto optimal ('coordinating') behavior and subsequent reconstruction of each radar's utility function, given a finite dataset of radar network emissions. The method for accomplishing this is derived from the micro-economic setting of Revealed Preferences, and also applies to more general problems of inverse detection and learning of multi-objective optimizing systems.   
\end{abstract}


\vspace{-0.3cm}
\section{Introduction}
\vspace{-0.2cm}
Cognitive radars \cite{haykin2006cognitive} use the perception-action cycle of cognition to sense the target, learn relevant information, then optimally adapt their output emissions in response. We consider the case when there is a \textit{network} of cognitive radars which collaborate to optimally track a target. In a coordinating radar network, not only do the individual cognitive radars optimally adapt their output (with respect to an individual utility function) subject to resource constraints, but also the \textit{allocation of resources} between radars is subject to an optimization procedure. The resource to be allocated is often interpreted as the total power available to the radar network at each time step. Such an optimal power allocation strategy has been studied in \cite{shi2017power}, \cite{panoui2014game}, \cite{chavali2011scheduling}, \cite{bacci2012game}, in which algorithmic game-theoretic methods are employed. Specifically, \cite{shi2017power} poses the problem of adaptive power allocation for radar networks as a cooperative game, and provides an iterative cooperative Nash bargaining algorithm which converges quickly to the Pareto optimal equilibrium.\\
\vspace{-0cm}
In this work we are interested in the \textit{inverse} problem. Suppose "we" are a target being tracked by a group of cognitive radars, and we have a mechanism for intercepting some of their emissions. Given such a set of measured signals, how can we determine if the radars compose a \textit{coordinating} (in the sense of Pareto optimal power allocation) network?  If they are coordinating, can we subsequently reconstruct the individual utility functions which induce each radar's output? This problem is similar to multi-agent inverse reinforcement learning (MAIRL) frameworks \cite{yu2019multi}, \cite{natarajan2010multi}, \cite{lin2019multi} in that we aim to reconstruct each radar's utility function. However, it is fundamentally different in that we first aim to \textit{detect coordination} among the radars; MAIRL approaches \textit{a priori} assume some structure, such as coordination, among the agents. Thus, this problem can be considered at the intersection of inverse game theory \cite{kuleshov2015inverse} and MAIRL. \\
\vspace{-0cm}
In \cite{krishnamurthy2020identifying}, the framework of \cite{afriat1967construction} was exploited to approach this problem for the single-radar case (i.e. detect if a single radar is a constrained utility maximizer). We extend this work to the multi-radar regime by employing
the micro-economic framework of \cite{cherchye2011revealed}. In \cite{cherchye2011revealed} a Revealed Preference methodology is applied to the setting of collective consumption models, and a framework is developed for testing consistency with Pareto optimality given a finite dataset of input-output group responses. Remarkably, conditions are given which are \textit{necessary and sufficient} for a finite dataset to be consistent with Pareto optimality. We map our radar network inverse problem into this framework, thus allowing us to test for Pareto optimal power allocation (coordination) given a finite dataset of (incompletely) measured cognitive radar network emissions. We also show how this can be used to reconstruct feasible utility functions for each observed radar. \\
\vspace{-0cm}
This paper is organized as follows: In section~\ref{sec:RP_analysis} we provide background on the micro-economic framework of \cite{cherchye2011revealed}. In section~\ref{sec:Rad_coord} we introduce the radar network protocol and measurement model, and show how this can be mapped into the framework presented in Section~\ref{sec:RP_analysis}. In section~\ref{sec:num_ex} we provide a numerical example that illustrates the validity of our coordination-testing and utility reconstruction procedures. Finally we conclude in section~\ref{sec:conc}. 

\vspace{-0.3cm}

\section{Coordination Detection and Utility Reconstruction}
\label{sec:RP_analysis}

Our framework for coordination detection and utility function reconstruction (inverse reinforcement learning) is the micro-economic setting of Revealed Preferences. 
Here we present the Revealed Preference framework for the collective consumption setting \cite{cherchye2011revealed}. We present this framework in its original generality, but in section~\ref{sec:Rad_coord} we adapt it to the radar network power allocation setting. 

\vspace{-0.2cm}

\subsection{Agent Consumption}
Consider a group composed of $M$ members, each of which can consume some quantity of $N$ goods. At time $t \in \nat$ each member $i \in \{1,\dots, M\}$ has consumption given by the vector $\consi \in \posreals^N$, and the aggregate group consumption is given by $\cons = \sum_{i=1}^M \consi \in \posreals^N$. It is assumed that the preferences of each member $i$ can be represented by a non-satiated and non-decreasing utility function $\utili(\beta),\ \beta \in \posreals^N$. If the agents choose their individual consumptions in accordance with group Pareto optimality, then we say the group \textit{coordinates}:

\begin{definition}
\label{coll_rat}
A group of M agents with respective utility functions $\{U^i: \reals^N \to \reals\}_{i=1}^M$, and subject to price vector $\pricet$, behaves with \textbf{coordination} if consumption $\{\consi\}_{i=1}^M$ satisfies
\vspace{-0.3cm}
\begin{align}
\begin{split}
\label{eq:coll_rat}
&\sum_{i=1}^M\mu_t^iU^i(\consi) \geq \sum_{i=1}^M\mu_t^iU^i(\zeta^i) \ \  \\&\forall \{\zeta^i \in \reals^N_+ \}_{i=1}^M: \  \pricet'(\sum_i\zeta^i) \leq \pricet'(\sum_i\consi) 
\end{split}
\end{align}
for some set of Pareto weights $\{\parweighti\}_{i=1}^M$
\end{definition}
The Pareto weights $\{\parweighti\}_{i=1}^M$ correspond to the bargaining power, or relative importance, of each individual in the formation of the Pareto optimal solution. We now switch perspectives to an empirical analyst who observes the group consumption and attempts to detect coordination, in the sense of Definition~\ref{coll_rat}. 

\subsection{Analyst Detection of Coordination}
For each $t \in \{1,\dots,T\}$, suppose an analyst observes probe signals $\pricet \in \reals^N_+$, aggregate consumption $\cons 
 = \sum_{i=1}^M\consi \in \reals^N$, and "assignable quantities" $\asconsi \leq \consi \ \forall i \in \{1,\dots,M\}$. The assignable quantity $\asconsi$ represents observed consumption of some subset of the total consumed quantities by individual $i$ at time $t$, and hence it is elementwise no greater than the actual consumption vector $\consi$. The apparatus by which these assignable quantities are observed, and the amount observed, varies by application. We denote this dataset as 
 \begin{equation}
 \label{dataset}
    \dataset = \{\pricet, \cons, \{\asconsi\}_{i=1}^M, \ t \in \{1,\dots,T\} \}
 \end{equation}
 We emphasize that the true individual consumption vectors $\consi$ are hidden. Given this dataset, \cite{cherchye2011revealed} provides necessary and sufficient conditions for consistency with coordination, in the sense of Definition~\ref{coll_rat}. To formulate these conditions, we need one more definition:
 For each observation $t$, \textit{feasible personalized quantities} $\feasconsi \in \posreals^N, \ i=1,\dots,M$ satisfy $\feasconsi \geq \asconsi \ \forall i$ and $\sum_{i=1}^M \feasconsi = \cons$.

 Now we provide the main result of \cite{cherchye2011revealed} which states the equivalence between a consistency of a dataset $\dataset$ with coordination and existence of a non-empty feasible region of a set of inequalities.

 \begin{theorem}{\cite{cherchye2011revealed}}
 \label{thm:cherchye1}
    Let $\dataset$ be a set of observations. The following are equivalent:
    \begin{enumerate}
    \item there exist a set of $M$ concave and continuous utility functions $U^1,\dots,U^m$ such that $\dataset$ is consistent with coordination, in the sense of Definition~\ref{coll_rat}.
    \vspace{-0.2cm}
    \item there exist feasible personalized quantities $\feasconsi$ and numbers $u_j^i > 0, \lambda_j^i > 0$ such that for all $s,t \in \{1,\dots,T\}$: 
    \vspace{-0.3cm}
    \begin{equation}
    \label{af_ineq}
        u_s^i - u_t^i \leq \lambda_t^i[\pricet'q_s^i - \pricet'q_t^i]
    \end{equation}
    for each member $i=1,\dots,M$.
    \end{enumerate}
\end{theorem}
\begin{proof}
See Appendix A.1 of \cite{cherchye2011revealed}.
\end{proof}
    \vspace{-0.2cm}
    \begin{lemma}
    \label{lem:util}
    Suppose the conditions of Theorem~\ref{thm:cherchye1} hold. Explicit monotone and concave utility functions that rationalize the dataset by satisfying the condition of coordination given in Definition~\ref{coll_rat} are given by
            \begin{equation}
            \label{eq:util}
                \hat{U}^i(\beta) = \min_{t \in \{1,\dots,N\}} \{u_t^i + \lambda_t^i \alpha_t'(\beta - q_t^i)\}
            \end{equation}
   \end{lemma}
 \begin{proof}
 \vspace{-0.2cm}
     This follows immediately by application of Afriat's Theorem \cite{afriat1967construction}.
 \end{proof}

 Thus if the feasible region of \eqref{af_ineq} is nonempty, then monotone concave rationalizing utility functions can be reconstructed using \eqref{eq:util}.

 Implementing this feasibility test \eqref{af_ineq} in practice presents a problem however, because the set $\{\feasconsi\}_{i=1}^M$ is \textit{unobserved}. To overcome this, \cite{cherchye2011revealed} present an equivalent mixed-integer linear program (MILP) formulation which allows for practical computation of this coordination test.

 \begin{theorem}{\cite{cherchye2011revealed}}
 \label{MILP}
 Let $\dataset$ be a set of observations. There exists a set of M concave and continuous utility functions $U^1,\dots,U^M$ such that $\dataset$ is consistent with coordination if and only if there exist $\feasconsi \in \posreals^N$, $\eta_t^i \in \posreals$, and $x_{st}^i \in \{0,1\}, i=1,\dots,M$ that satisfy
 \vspace{-0.1cm}
 \begin{enumerate}[label=\roman*)]
    \item $\sum_{i=1}^M\feasconsi = \cons$ and $\feasconsi \geq \asconsi$
    \vspace{-0.2cm}
    \item $\eta_t^i = \pricet'\feasconsi$
    \vspace{-0.2cm}
    \item $\eta_s^i - \pricet'\feasconsi < y_s x_{st}^i$
    \vspace{-0.2cm}
    \item $x_{su}^i + x_{ut}^i \leq 1 + x_{st}^i$
    \vspace{-0.2cm}
    \item $\eta_t^i - \pricet'\feasconsi \leq y_t(1-x_{st}^i)$
 \end{enumerate}
 \vspace{-0.1cm}
 where $y_t = \pricet'\cons$ is the group consumption at time $t$.

 \end{theorem}

 \begin{proof}
    See Proposition 2 of \cite{cherchye2011revealed}.
 \end{proof}

\section{Testing for Cognitive Radar Network Coordination}
\label{sec:Rad_coord}
Here we test for Pareto optimal power allocation among $M$ radars, and subsequently reconstruct each radar's utility function, from a finite dataset of intercepted radar emissions. We take a target-centric view and assume that the target can measure the power of the radar network signals that are being output in response to its maneuvers. 
\begin{figure}
\centering
  \includegraphics[width=\linewidth,scale=0.2]{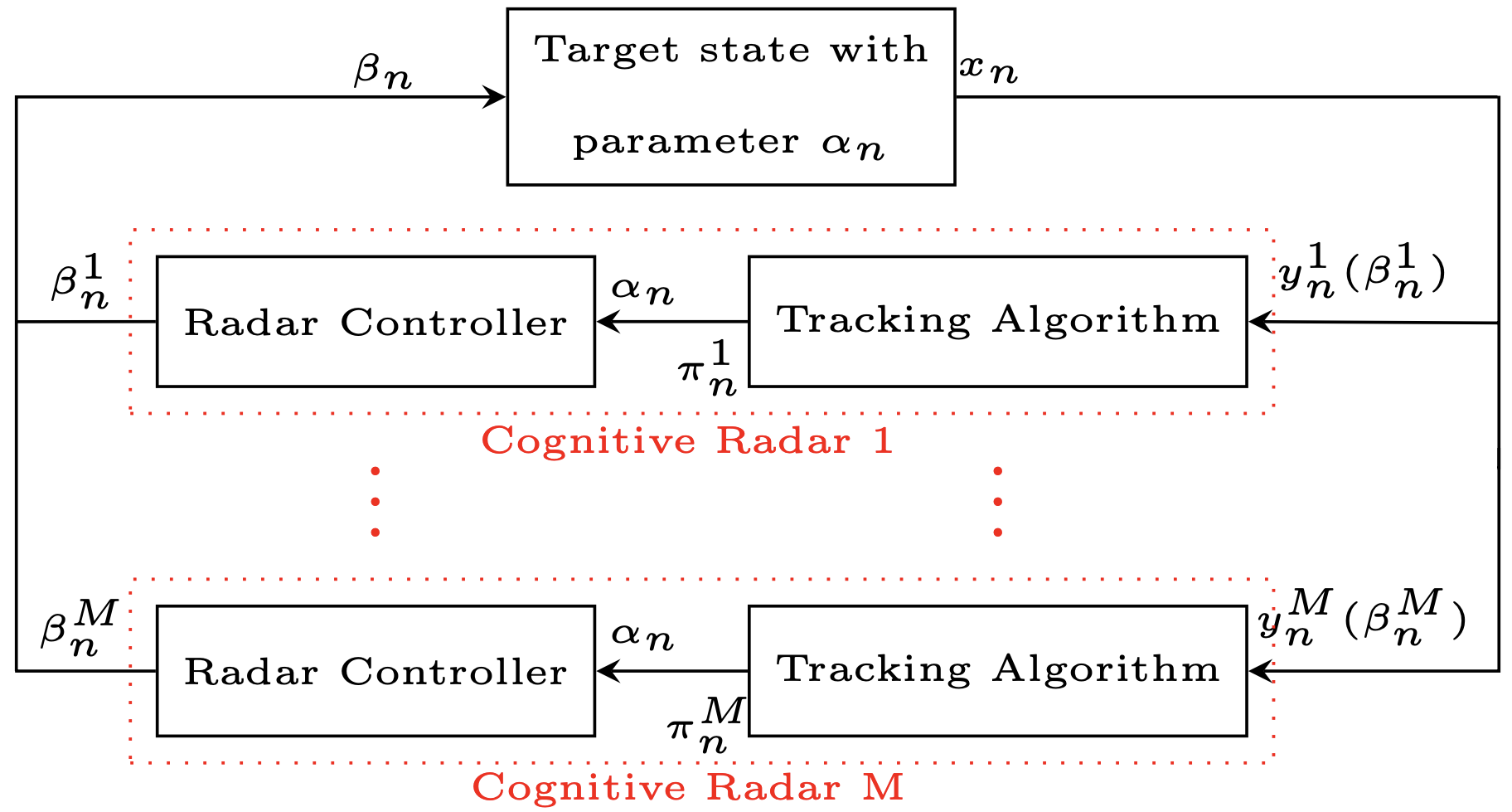}
  \caption{\scriptsize Interaction of our dynamics with the adversary's cognitive radar network. Each cognitive radar is comprised of a Bayesian tracker and a radar controller. Based on the time series $\{\alpha_n,  \beta_n, \hat{\beta}_n^1, \dots,\hat{\beta}_n^M\}_{n=1}^T$, our goal is to determine if the radar network is coordinating, i.e., if the true responses $\{\beta_n^i, n=1,\dots,T\}_{i=1}^M$ satisfy coordination as in \eqref{eq:coll_rat}.}
  \label{fig:blockdiag}
\end{figure}

\vspace{-0.3cm}
\subsection{Waveform Adaptation by Cognitive Radar Network}

\vspace{-0.2cm}
\subsubsection{Cognitive Radar Network Response Model}
\vspace{-0.2cm}
Here we outline the abstract interaction between "us" (target) and the network of $M$ cognitive radars. Let $n=1,2,\dots$  denote discrete time. Our probe signal (modeling acceleration maneuvers) is $\alpha_n \in \reals^N$ and the response of radar $i$ is $\beta_n^i \in \reals^N$. The Pareto optimal interaction dynamics are as follows:
\begin{enumerate}[label=\roman*)]
\vspace{-0.2cm}
\item our state $x_n \sim p_{\alpha_n}(x|x_{n-1}), x_0 \sim \pi_0$
\vspace{-0.2cm}
\item radar $i$ observation $y_n^i \sim p_{\beta_n^i}(y|x_n)$
\vspace{-0.2cm}
\item radar $i$ tracker $\pi_n^i = T(\pi_{n-1}^i, y_n^i)$ 
\vspace{-0.2cm}
\item radar $i$ output $\beta_n^i$, where 
\vspace{-0.5cm}
 \begin{equation}
\label{MOO}
\{\beta_n^i\}_{i=1}^M \in \textrm{argmax}_{\{\zeta_i\}_{i=1}^M : \alpha_n (\sum_{i=1}^M \zeta_i) \leq C} \sum_{i=1}^M \mu_n^i U^i(\zeta_i)
\end{equation}
and $C$ is the radar network power constraint
\vspace{-0.2cm}
\item aggregate radar network output $\beta_n = \sum_{i=1}^m \beta_n^i$
\end{enumerate}
These interaction dynamics can be seen in Figure~\ref{fig:blockdiag}. The interpretation of $\alpha_n$ is as a parameter which influences the state transition dynamics of the target (us). This will soon be specified in terms of the eigenvalue spectrum of the state-noise covariance matrix. The coordinating radar network output $\{\beta_n^i\}_{i=1}^M$ lies on the Pareto Front corresponding to the multi-objective optimization problem $\textrm{argmax}_{\{\zeta_i\}_{i=1}^M} \sum_{i=1}^M \mu_n^i U^i(\zeta_i)$, with constraint $\alpha_n (\sum_{i=1}^M \zeta_i) \leq C$ and Pareto weights $\{\mu_n^i\}_{i=1}^M$. At an abstract level, the observations of the $i$'th radar, $y_n^i$, are parameterized by its output signal $\beta_n^i$. We will soon specify this interpretation in terms of the eigenvalue spectrum of the measurement noise covariance matrix. 
\vspace{-0.3cm}
\subsubsection{Linear Gaussian Target Model and Radar Tracker}
Linear Gaussian dynamics for a target's kinematics \cite{li2003survey} and linear Gaussian measurements at each radar are widely assumed as a useful approximation \cite{bar2004estimation}. Thus we will consider the following linear Gaussian state dynamics and measurements, as a special case of the general network response model of the previous subsection:
\begin{align}
    \begin{split}
    \label{lin_gaus}
        & x_{n+1} = A x_n + w_n(\alpha_n), \  x_0 \sim \pi_0 \\
        & y_n^i = C x_n + v_n^i(\beta_n^i) 
    \end{split}
\end{align}
Here we represent the target (our) state as $x_n \in \reals^N$, with initial density $\pi_0 \sim \gaussN(\hat{x}_0,\Sigma_0)$. Cognitive radar $i$ obtains measurements $y_n^i \in \reals^N$. State and measurement noise processes are given by $w_n \sim \gaussN(0,Q(\alpha_n)), v_n^i \sim \gaussN(0,R(\beta_n^i))$, respectively, and we assume $\{w_n\}, \{v_n^1\},\dots,\{v_n^M\}$ are independent i.i.d. processes. Note that we implement the dependence of the dynamics on signal $\alpha_n$ and response $\{\beta_n^i\}$ through the state and observation noise covariance matrices $Q$ and $R$, respectively. A radar's control of its ambiguity function can be mapped to a control of the measurement noise covariance matrix $R$. However, reducing the measurement noise covariance incurs several costs to the system, such as increased power output. In our setting the cost constraint is interpreted as a bound on the total available power to be allocated across the $M$ radars. The target's probing of the radars is performed via purposeful maneuvers that result in modulation of the state noise covariance matrix $Q$ in \eqref{lin_gaus} by $\alpha_n$.

\vspace{-0.3cm}

\subsubsection{Testing for Coordination among Cognitive Radars}
We now show how Theorem~\ref{thm:cherchye1} can be used to determine if the cognitive radar network is coordinating, i.e. the responses satisfy coordination in the sense of Definition~\eqref{coll_rat}. What remains is to specify the structure of the process by which we (the target) can acquire information about the radar network signal output. Specifically, suppose:
\begin{enumerate}[label=\roman*)]
    \vspace{-0.25cm}
    \item Our probe $\alpha_n$, characterizing our maneuvers, is the vector of eigenvalues of the positive definite matrix $Q$ governing the state-noise covariance in \eqref{lin_gaus}.
    \vspace{-0.25cm}
    \item The response $\beta_n^i$ of the $i$'th radar is the vector of eigenvalues of the positive definite matrix $R^{-1}$ governing the noise covariance in \eqref{lin_gaus}.
    \vspace{-0.25cm}
    \item $\{\beta_n^i\}_{i=1}^M$ is chosen to be Pareto optimal with respect to individual utilities $\{U^i\}_{i=1}^M$ and Pareto weights $\{\mu_n^i\}_{i=1}^M$, and subject to the constraint $\alpha_n (\sum_{i=1}^M \beta_n^i) \leq 1$, i.e. chosen according to \eqref{MOO}. 
    \vspace{-0.25cm}
    \item We (the target) can observe the aggregate output signal power across all $M$ radars $\beta_n = \sum_{i=1}^M \beta_n^i$ through a broadbeam omni-directional receiver configuration.
    \vspace{-0.25cm}
    \item For each radar we observe an associated signal power $\hat{\beta}_n^i$ which is upper-bounded by the true signal power $\beta_n^i$, i.e. $\hat{\beta}_n^i \leq \beta_n^i$ . This is a reasonable assumption that can arise from a  narrowbeam receiver that steers to one radar at a time, and thus only observes some fraction of the power output of each radar.
\end{enumerate}
\vspace{-0.2cm}
Thus we obtain the dataset 
\vspace{-0.2cm}
\begin{equation}
\label{obs_data}
    \dataset = \{\alpha_n, \beta_n , \{\hat{\beta}_n^i\}_{i=1}^M, n=1,\dots,T\}
\end{equation}
with $\alpha_n, \beta_n, \hat{\beta}_n^i \in \reals^N$. Then we can directly use the MILP from Theorem~\ref{MILP} to test for coordination, and if this MILP has a non-empty feasible region then we can employ Lemma~\ref{lem:util} to reconstruct feasible utility functions for each of the $M$ cognitive radars.

\vspace{-0.3cm}

\section{Numerical Example: Tri-Radar Spectral Coordination Test}
\label{sec:num_ex}
\vspace{-0.3cm}
Here we demonstrate the use of the MILP of Theorem \ref{MILP} to test for Pareto optimal power allocation of a network of $M=3$ radars. For ease of visualization, we choose $N = 2$ so that the probe and response signals $\alpha_n, \beta_n^i \in \reals^2$ and thus the utility functions $U^i(\cdot)$ can be displayed on a 2-d contour plot. The elements of our probe signal $\alpha_n$ are chosen randomly and independently over time $n$ as $\alpha_n(1) \sim \textrm{Unif}(0.1,1.1)$ and $\alpha_n(2) \sim \textrm{Unif}(0.1,1.1)$ where Unif$(a,b)$ denotes the uniform pdf with support $(a,b)$. Our probe signal parameterizes the state-noise covariance matrix as $Q_n = \textrm{diag}[\alpha_n(1),\alpha_n(2)]$ in \eqref{lin_gaus}. Suppose the utility functions of the radars are $U^1(\beta) = \textrm{det}(R^{-1}(\beta)) = \beta(1) \times \beta(2), \ U^2(\beta) = \textrm{Tr}(R^{-1}(\beta)) = \beta(1) + \beta(2), \ U^3(\beta) = \sqrt{\beta(1)} \times \beta(2)$.  The distributed response $\{\beta_n^i\}_{i=1}^3$ is then chosen at each time $n$ such that \eqref{MOO} is satisfied, with $M=3$, $C=1$, $\mu_n^1 = 0.3, \mu_n^2 = 0.3, \mu_n^3 = 0.4 \ \forall n \in \{1,\dots,T\}$. This simulation is run for $T=10$ steps. 

We simulate the sensor-level measurement to obtain the dataset \eqref{dataset} in the following way: the aggregate consumption $\beta_n = \sum_{i=1}^3 \beta_n^i$ is directly measured for all $n=1,\dots,T$, and the assignable quantities are obtained as $\beta_n^i = S \times \beta_n^i,$ where $S \sim \textrm{Unif}(0.1,1)$. Using this dataset, we have verified\footnote{Matlab's intlinprog function is used to test the linear program} that the MILP of Theorem~\ref{MILP} has a feasible solution\footnote{There exists $\{q_t^i \in \reals^N_+, \eta_t^i \in \reals_+, x_{st}^i \in \{0,1\}, i=1,2,3\}_{s,t=1}^{10}$ such that conditions $i-v$ of Theorem~\ref{MILP} hold.}, indicating that the data is consistent with Pareto optimality, or coordination. Then taking a resultant feasible set $\{q_t^i, i=1,2,3\}_{t=1}^{10}$, we use \eqref{eq:util} to reconstruct feasible utility functions which rationalize the data. These reconstructed utility functions are shown in Fig.~\ref{fig:utilities}, where plots~\ref{fig:T_util1}, \ref{fig:T_util2}, \ref{fig:T_util3} are the true utility functions of each radar and plots~\ref{fig:util1}, \ref{fig:util2}, \ref{fig:util3} are the reconstructed utility functions for each respective radar, given the dataset $\dataset$. \\
\vspace{-0cm}
Conversely, we have verified that the test of Theorem~\ref{MILP} is capable of rejecting non-cooperative data. Specifically, we tested the MILP 100 times with each $\beta_n^i$ generated independently (without coordination) as a two-dimensional uniform random variable, $\beta_n^i = [\beta_n^i(1),\beta_n^i(2)]$ where $\beta_n^i(k) \sim \textrm{Unif}(0,1), k=1,2$. Out of these 100 tests, 5 resulted in a non-empty feasible region, indicating coordination is present, and 95 resulted in an empty feasible region, indicating the absence of coordination. Thus we observe a Type-II error (conclude that the radars are coordinating when in fact they are acting independently) rate of 5\%.

\begin{figure}[t]
        \centering
        \begin{subfigure}[]{0.23\textwidth} 
            \centering
            \includegraphics[width=\textwidth]{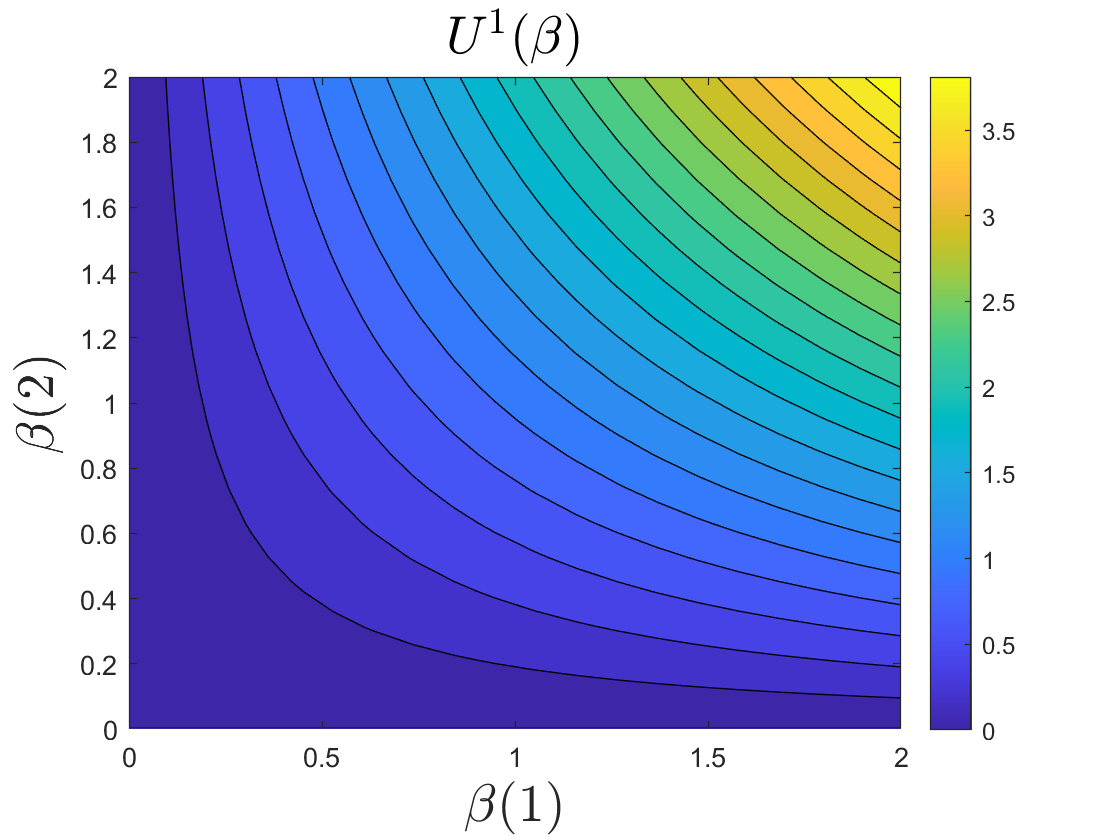}
            \caption[]%
            {{\scriptsize $U^1(\beta) = \textrm{det}(R^{-1}(\beta))$}}    
            \label{fig:T_util1}
        \end{subfigure}
        \hfill
        \begin{subfigure}[]{0.23\textwidth}  
            \centering 
            \includegraphics[width=\textwidth]{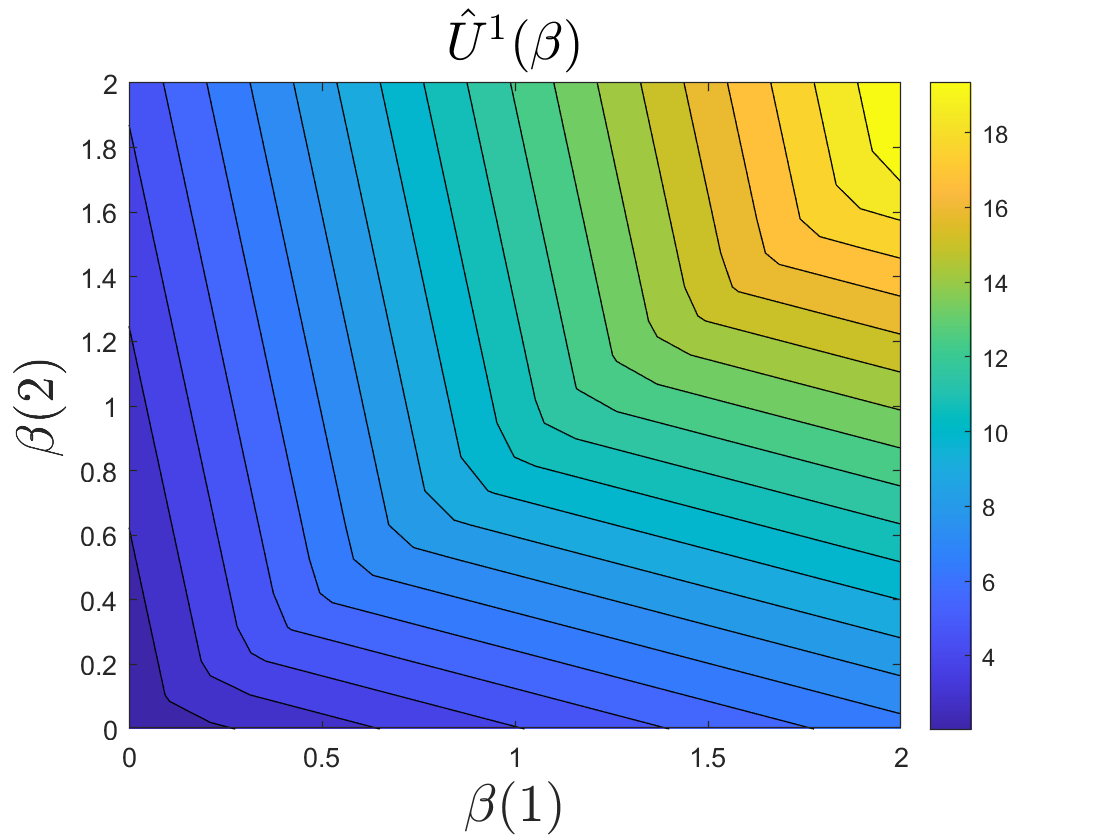}
            \caption[]%
            {{\scriptsize $\hat{U}^1(\beta)$}}    
            \label{fig:util1}
        \end{subfigure}
        \vskip\baselineskip

        \begin{subfigure}[]{0.23\textwidth}   
            \centering 
            \includegraphics[width=\textwidth]{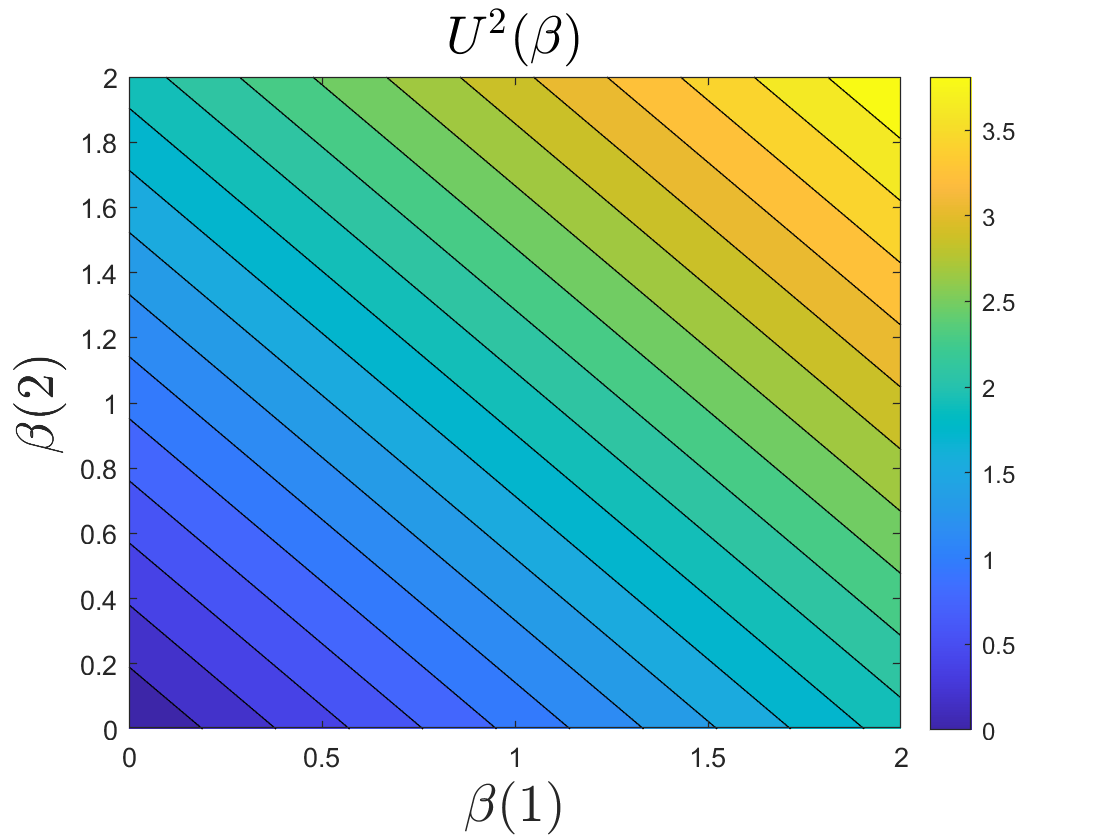}
            \caption[]%
            {{\scriptsize $U^2(\beta) = \textrm{Tr}(R^{-1}(\beta))$}}    
            \label{fig:T_util2}
        \end{subfigure}
        \hfill
        \begin{subfigure}[]{0.23\textwidth}   
            \centering 
            \includegraphics[width=\textwidth]{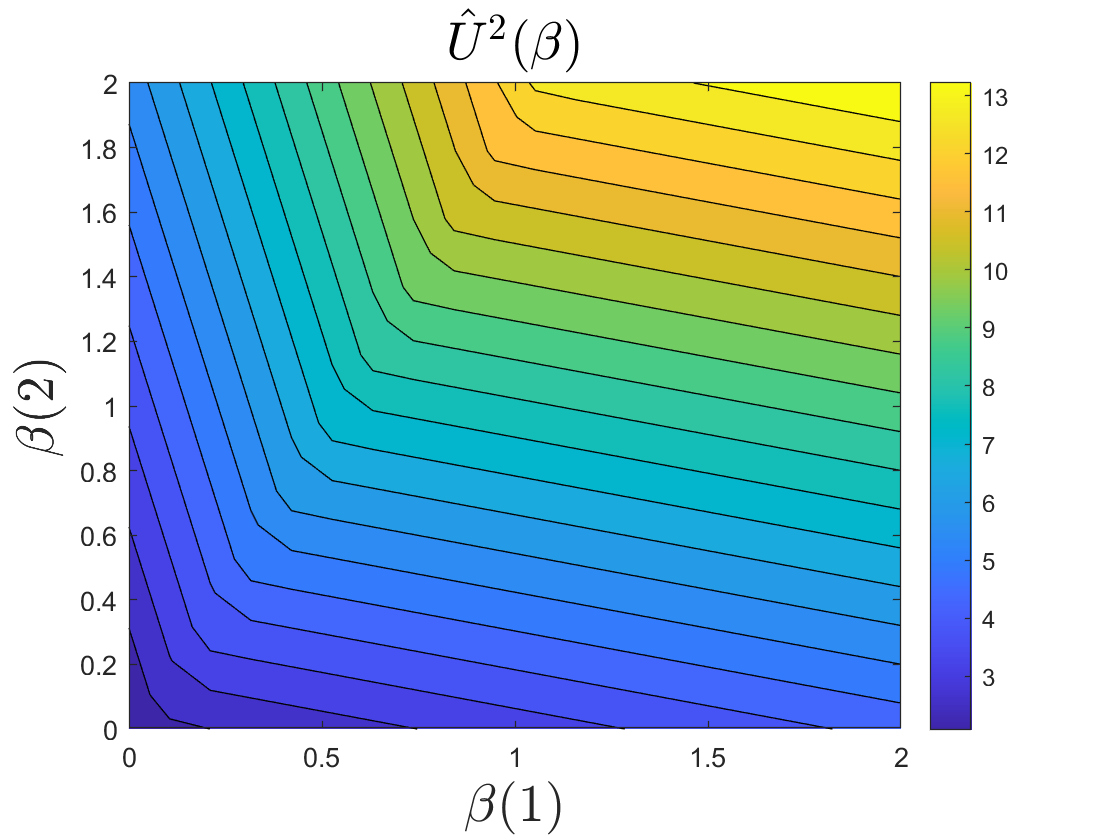}
            \caption[]%
            {{\scriptsize $\hat{U}^2(\beta)$}}    
            \label{fig:util2}
        \end{subfigure}
        \vskip\baselineskip

        \begin{subfigure}[]{0.23\textwidth}   
            \centering 
            \includegraphics[width=\textwidth]{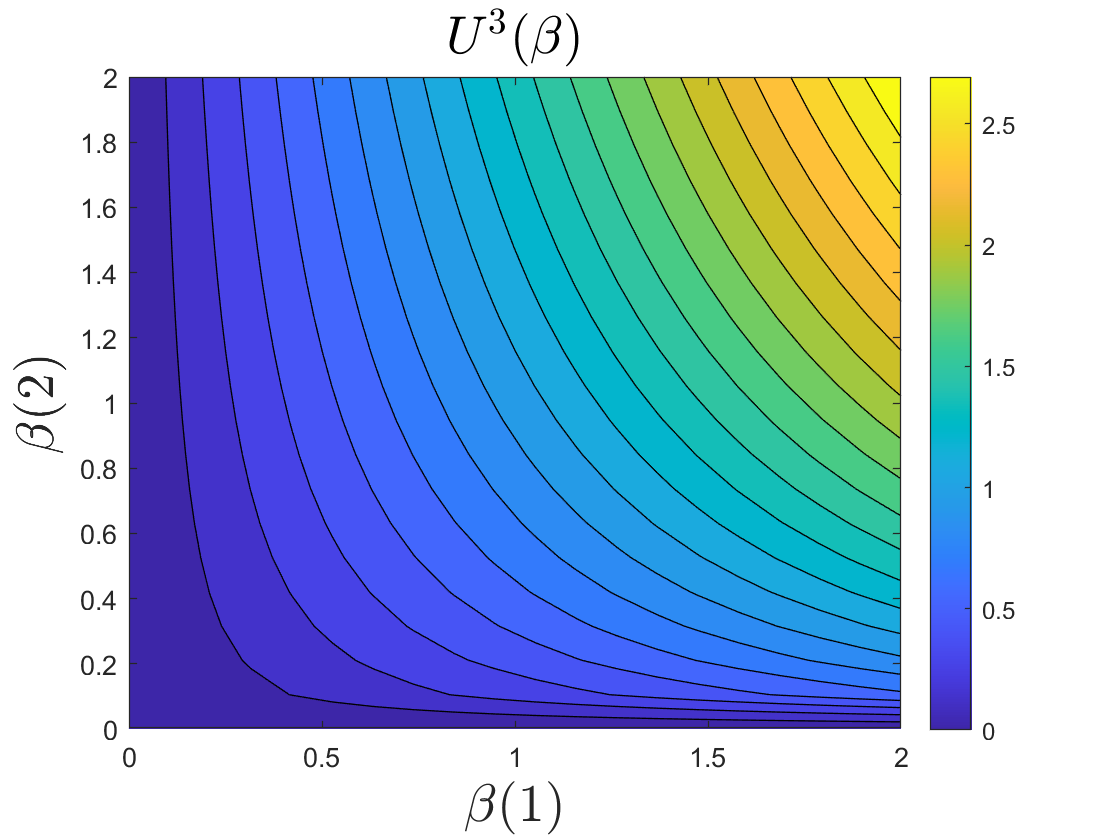}
            \caption[]%
            {{\scriptsize $U^3(\beta) = \sqrt{\beta(1)}\beta(2)$}}    
            \label{fig:T_util3}
        \end{subfigure}
        \hfill
        \begin{subfigure}[]{0.23\textwidth}   
            \centering 
            \includegraphics[width=\textwidth]{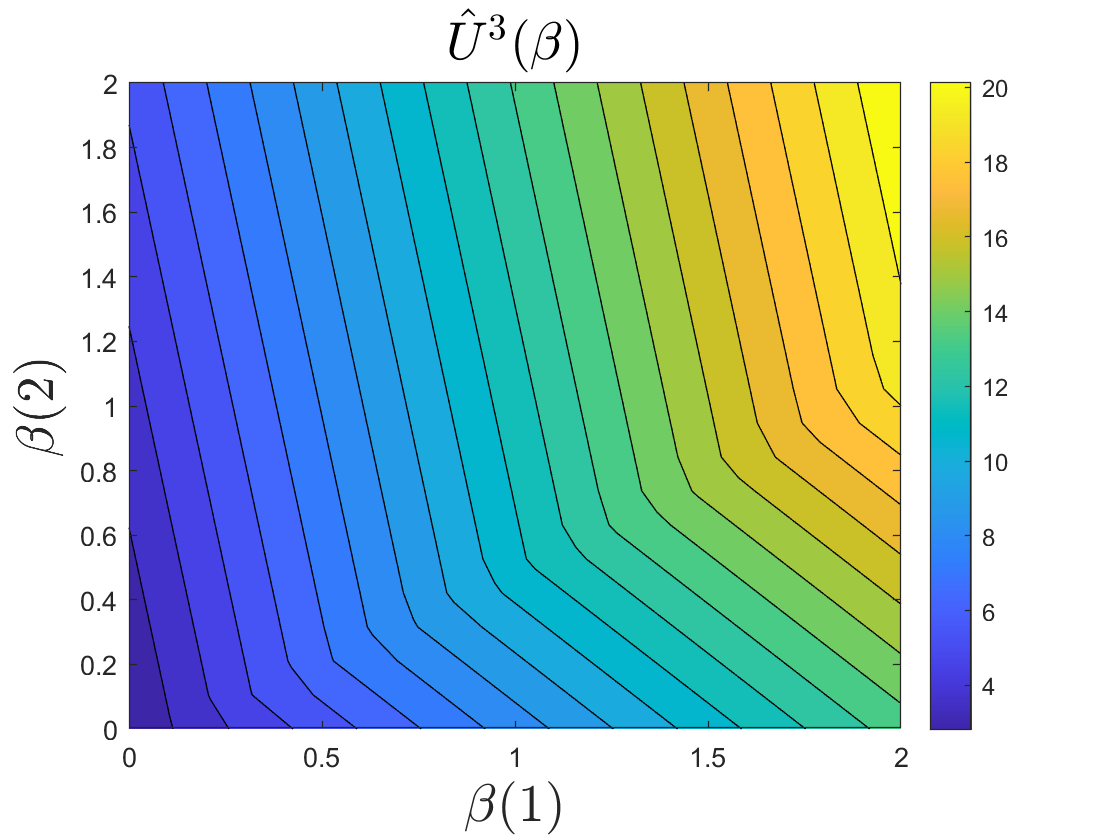}
            \caption[]%
            {{\scriptsize $\hat{U}^3(\beta)$}}    
            \label{fig:util3}
        \end{subfigure}
        \caption[]
        {\scriptsize $U^i(\beta)$ is the true utility function of the $i$'th radar, inducing the responses $\{\hat{\beta}_n^i\}_{n=1}^{10}$. Given the dataset $\dataset = \{\alpha_n, \hat{\beta}_n, \{\hat{\beta}_n^i\}_{i=1}^3\}_{n=1}^{10}$, the MILP \ref{MILP} has a non-empty feasible region. From this MILP we obtain the set $\{q_t^i\}_{t=1}^{10}$ and reconstruct $\hat{U}^i(\beta)$ via \eqref{eq:util} for each $i=1,2,3$.}
        \label{fig:utilities}
\end{figure}

\vspace{-0.5cm}

\section{Conclusion}
\vspace{-0.3cm}
We have utilized the micro-economic framework of Revealed Preferences to test for Pareto optimality, indicating coordination, among a cognitive radar network. We have provided a numerical example which validates the proposed framework and illustrates the reconstructed utility functions of the individual cognitive radars.

\label{sec:conc}

\bibliographystyle{IEEEtran}
\bibliography{Bibliography.bib}

\end{document}